\documentclass[12pt, a4paper]{article}

\usepackage{amsmath}
\usepackage{amsopn}
\usepackage{amsfonts, mathrsfs}
\usepackage{amsthm}
\usepackage{epsfig, psfrag}

\numberwithin{equation}{section}

\theoremstyle{remark}
\newtheorem{remark}{Remark}
\newtheoremstyle{mytheorem}{0.5cm}{0.2cm}{\slshape}{ }{\bfseries}{.}{ }{}
\theoremstyle{mytheorem}
\newtheorem{theorem}{Theorem}[section]

\newtheorem{corollary}[theorem]{Corollary}

\newcommand{\ti}{{\to \infty}}
\newcommand{\Prob}[1]{\mathbf{P}\left\{#1\right\}}
\newcommand{\E}{\mathbf{E}}
\newcommand{\Ind}{\mathbf{1}}

\newcommand{\standardspace}[1]{\mathbb{#1}}
\newcommand{\R}{\standardspace{R}}
\newcommand{\Sphere}{{\standardspace{S}^{d-1}}}

\newcommand{\Ball}{{\standardspace{B}^d}}

\newcommand{\thf}{\frac{1}{2}}

\newcommand{\Ht}{\tilde{H}}
\newcommand{\D}{\Delta_{nt}}

\begin{document}

\bibliographystyle{plain}

\title{Asymptotic Normality of \\$U$-Quantile-Statistics}

\author{Michael Mayer\footnote{Department of Mathematical
    Statistics and Actuarial Science, University of Bern,
    Alpeneggstrasse 22, CH-3012 Bern, Switzerland.}}
\date{}
\bibliographystyle{plain}
\maketitle

\begin{abstract}
    In 1948, W. Hoeffding introduced a large class of unbiased
    estimators called $U$-statistics, defined as the average
    value of a real-valued $m$-variate function $h$ calculated at all
    possible sets of $m$ points from a random sample.
    In the present paper, we investigate the corresponding robust
    analogue which we call $U$-quantile-statistics.
    We are concerned with the asymptotic behavior of the
    sample $p$-quantile of such function $h$ instead of its average.
    Alternatively, $U$-quantile-statistics can be viewed as quantile
    estimators for a certain class of dependent random variables.
    Examples are given by a slightly modified Hodges-Lehmann estimator of location
    and the median interpoint distance among random points in space.


\end{abstract}

\section{Introduction}
$U$-statistics form a very important class of unbiased estimators
for distributional properties such as moments or Spearman's rank
correlation. A $U$-statistic of degree $m$ with symmetric kernel $h$
is a function of the form \linebreak
\begin{equation}\label{eq:u-statisic}
    U_n(\xi_1, \dots, \xi_n) = {n \choose m}^{-1} \sum_J h(\xi_{i_1}, \cdots,
    \xi_{i_m}),
\end{equation}
where the sum is over $J = \{(i_1, \dots, i_m)\!: 1 \leq i_1 <\dots
< i_m \leq n\}$, $\xi_1, \dots, \xi_n$ are random elements in a
measurable space $\mathcal{S}$, and $h$ is a real-valued Borel
function on $\mathcal{S}^m$, symmetric in its $m$ arguments. In his
seminal paper, Hoeffding~\cite{hoeffding48} defined $U$-statistics
for not necessarily symmetric kernels and for random points in
$d$-dimensional Euclidean space $\R^d$. Later the concept was
extended to arbitrary measurable spaces. Since 1948, most of the
classical asymptotic results for sums of i.i.d.\ random variables
have been formulated in the setting of $U$-statistics, such as
central limit laws, strong laws of large numbers, Berry-Ess\'{e}en
type bounds, and laws of the iterated logarithm.

In this article we replace the average in (\ref{eq:u-statisic})
by the sample $p$th quantile $\Ht_{pn}$ and study its asymptotic distribution.
By e.g.~replacing the average by the median ($p = 1/2$), ordinary $U$-statistics
are robustified in a natural way.

For any distribution function $F$, the
$p$th quantile, $0 < p < 1$, is given by
\begin{displaymath}
    \Ht_{p} = \inf\{x: \ F(x) \geq p\},
\end{displaymath}
which satisfies the inequality
\begin{displaymath}
    F(\Ht_{p}-) \leq p \leq F(\Ht_{p}).
\end{displaymath}
Here, the sample $p$th quantile $\Ht_{pn}$ is defined as the $p$th quantile of the empirical distribution function of the sequence
of dependent random variables
\begin{equation}\label{eq:sequence}
    \left\{h(\xi_{i_1}, \dots, \xi_{i_m}), \quad 1 \leq i_1 < i_2 < \dots < i_m \leq n\right\},
\end{equation}
i.e.~$\Ht_{pn}$ is a value that separates the lowest $100p\%$ random variables in (\ref{eq:sequence}) from the rest.

Under mild smoothness conditions on the distribution function $F$ of $h(\xi_{i_1}, \dots, \xi_{i_m})$, we proof asymptotic normality for this class of estimators for $0 < p < 1$. The exceptions $p = 0$ and $p = 1$, corresponding to the extreme values of the dependent sequence (\ref{eq:sequence}), were already investigated in Lao and Mayer \cite{mayer08}. For bounded kernels, they established Weibull limit laws for these so called $U$-max-statistics. Their results are mainly based on a Poisson approximation theorem for $U$-statistics, see e.g. Barbour et al. \cite{bar:hol:jan92}.

In Section~\ref{se:theorem} we present the main result of the article and discuss asymptotic relative efficiency of a general $U$-quantile-statistic with respect to the corresponding ordinary $U$-statistic. The proof of the main result is shown in Section~\ref{se:proof}.
In Section~\ref{se:example} we apply our results to show asymptotic normality for both a modification and a generalization of the well-known Hodges-Lehmann estimator of location. As a second application, we describe the limiting behavior of the median interpoint distance among a random sample of points in Euclidean space.

\newpage
\section{Asymptotic normality} \label{se:theorem}

Asymptotic normality of $\Ht_{pn}$ is stated in the main result of this article.

\begin{theorem}\label{th:main}
    Let $\xi_1, \dots, \xi_n$ be i.i.d.\ $\mathcal{S}$-valued random elements
    and \linebreak $h\!: \mathcal S^m \rightarrow \R$ a symmetric Borel
    function. Assume that the distribution function $F$ of $h(\xi_1, \dots, \xi_m)$
    is continuous at $\Ht_p$. Left- and right-hand derivatives of $F$ at $\Ht_p$ are denoted by
    $F'(\Ht_p-)$ and $F'(\Ht_p+)$, respectively, provided they exist. Put
    \begin{equation}\label{eq:zeta}
         \zeta = \Prob{h(\xi_1, \dots, \xi_m) \leq {\tilde H_p}, h(\xi_1, \xi_{m+1}, \dots, \xi_{2m-1}) \leq {\tilde H_p}} - p^2.
    \end{equation}
    Then, for $0 < p < 1$ and $\zeta > 0$,
    \begin{enumerate}
        \item[(i)] If there exists $F'(\Ht_p-) > 0$, then for $t < 0$,
            \begin{displaymath}
                \lim_{n\ti} \Prob{\frac{n^\thf(\Ht_{pn} - \Ht_p)}{m \zeta^\thf/F'(\Ht_p-)} \leq t} = \Phi(t).
            \end{displaymath}
        \item[(ii)] If there exists $F'(\Ht_p+) > 0$, then for $t > 0$,
            \begin{displaymath}
                \lim_{n\ti} \Prob{\frac{n^\thf(\Ht_{pn} - \Ht_p)}{m \zeta^\thf/F'(\Ht_p+)} \leq t} = \Phi(t).
            \end{displaymath}
    \end{enumerate}
\end{theorem}

As an immediate consequence of Theorem~\ref{th:main}, the following result holds.
\begin{corollary}\label{co:main}
    If $F$ in Theorem \ref{th:main} possesses a density $f$ in a neighborhood of $\Ht_p$ and $f$ is positive and continuous at $\Ht_p$, then
    \begin{displaymath}
        \lim_{n\ti} \Prob{\frac{n^\thf(\Ht_{pn} - \Ht_p)}{m \zeta^\thf/f(\Ht_p)} \leq t} = \Phi(t).
    \end{displaymath}
\end{corollary}
The continuity assumption for $f$ is required, as otherwise, $f$ could differ from $F'$ on a set with 0 mass.

\begin{remark} \label{re:conditioning}
    For $\mathcal{S} = \R$, by conditioning on the common random element $\xi_1$, $\zeta$ of (\ref{eq:zeta}) can be written as
    \begin{displaymath}
        \zeta = \int \left(\Prob{h(x, \xi_2, \dots, \xi_m) \leq {\tilde H_p}}\right)^2 dG(x) - p^2,
    \end{displaymath}
    where $G$ is the distribution function of $\xi_1$.
\end{remark}

\begin{remark}
    By setting $m = 1$, i.e. if $h$ is a function from $\mathcal{S}$ to $\R$, Theorem~\ref{th:main} implies Theorem A (p.~77) of Serfling~\cite{serfling80} on the asymptotic normality of the sample quantiles for i.i.d.~random variables.
\end{remark}

\begin{remark}
    By the central limit theorem for ordinary $U$-statistics (see e.g.~Hoeffding~\cite{hoeffding48} or Serfling~\cite{serfling80}), we are able to compare the asymptotic efficiency of the (robust) $U$-quantile-statistic (for $p = \thf$) with the asymptotic efficiency of $U_n$ given by~(\ref{eq:u-statisic}). Assume that the assumptions of Corollary~\ref{co:main} are fulfilled. Furthermore,
    assume that the density $f$ is symmetric about $\mu$, the random variables in~(\ref{eq:sequence}) have finite variance and
    \begin{displaymath}
        \zeta_1 = \E\left((h(\xi_1, \dots, \xi_m) - \mu) (h(\xi_1, \xi_{m+1}, \dots, \xi_{2m-1}) - \mu)\right) > 0.
    \end{displaymath}
    Then, the ordinary $U$-statistic $U_n$ based on kernel $h$ is asymptotically normal with variance $m^2\zeta_1$. Hence, by Corollary~\ref{co:main},
    \begin{displaymath}
        e(\Ht_{pn}, U_n) = f^2(\mu) \zeta_1 / \zeta.
    \end{displaymath}
\end{remark}

\section{Proof of Theorem \ref{th:main}} \label{se:proof}
We follow the proof of asymptotic normality of the usual $p$th quantile by Serfling (p.~78f) \cite{serfling80}, with the necessary adaptions.

\begin{proof}[Proof of Theorem \ref{th:main}]
    For fixed $t$ write
    \begin{align}
        G_n(t) &= \Prob{\frac{n^\thf\left(\Ht_{pn}-\Ht_p\right)}{A} \leq t} \nonumber\\
               &= \Prob{\Ht_{pn} \leq \Ht_p + t A n^{-\thf}} \nonumber\\
               &= \Prob{p \leq U_n(\D)} \label{eq:start},
    \end{align}
    where $A$ is a constant specified later and
    \begin{displaymath}
        U_n(\D) = {n \choose m}^{-1}\sum_{i_1 < \dots < i_m} \Ind\left\{h(\xi_{i_1}, \dots, \xi_{i_m}) \leq \Ht_{p} + t A n^{-\thf}\right\}
    \end{displaymath}
    is an ordinary $U$-statistic with expectation
    \begin{displaymath}
        \D = \Prob{h(\xi_{i_1}, \dots, \xi_{i_m}) \leq \Ht_{p} + t A n^{-\thf}}.
    \end{displaymath}
    By continuity of $F$ at $\Ht_p$, $\D \to p$ as $n \ti$.

    Since the kernel of $U_n(\D)$ is either 0 or 1, the third absolute moment $\lambda$ of $U_n(\D)$ exists. Furthermore,
    by continuity of probability functions (see e.g. p.~351 in Serfling \cite{serfling80}) and continuity of $F$ at $\Ht_p$, the quantity
     \begin{align*}
        \zeta_n &= \\
        &\Prob{h(\xi_1, \dots, \xi_m) \leq \Ht_p+tAn^{-\thf}, h(\xi_1, \xi_{m+1}, \dots, \xi_{2m-1}) \leq \Ht_p+tAn^{-\thf}} - \D^2
    \end{align*}
    converges to its limit $\zeta > 0$ as $n\ti$. Thus, for the normalized $U$-statistic
    \begin{displaymath}
        U^*_n(\D) = \frac{n^\thf \left(U_n(\D) - \D \right)}{m \zeta_n^\thf},
    \end{displaymath}
    by the Berry-Ess\'{e}en theorem for $U$-statistics by Callaert and Janson \cite{cal:jan78},
    \begin{equation}\label{eq:berry}
        \sup_{t \in \R} \left|\Prob{U^*_n(\D) \leq t} - \Phi(t) \right| \leq \frac{C\lambda}{n^\thf m^3\zeta_n^{3/2}}
    \end{equation}
    holds at least asymptotically as $n\ti$ for an universal constant $0 < C < \infty$.

    From (\ref{eq:start}) it follows that
    \begin{align*}
        G_n(t) &= \Prob{\frac{n^\thf\left(p-\D\right)}{m \zeta_n^\thf} \leq U^*_n(\D)} \\
                &= \Prob{U^*_n(\D) \geq -c_{nt}}
    \end{align*}
    with
    \begin{displaymath}
        c_{nt} = \frac{n^\thf\left(\D - p \right)}{m \zeta_n^\thf}.
    \end{displaymath}

    Clearly,
    \begin{align*}
        \Phi(t) - G_n(t) &= \Prob{U^*_n(\D) < - c_{nt}} - (1 - \Phi(t)) \\
                        & = \Prob{U^*_n(\D) < - c_{nt}} - \Phi(-c_{nt}) + \Phi(t) - \Phi(c_{nt}),
    \end{align*}
    and thus, by using the Berry-Ess\'{e}en bound (\ref{eq:berry}),
    \begin{displaymath}
        |G_n(t)-\Phi(t)| \leq \frac{C \lambda}{n^\thf m \zeta_n^{3/2}} + |\Phi(t) - \Phi(c_{nt})|.
    \end{displaymath}
    The first term on the right hand side vanishes as $n \ti$. It thus remains to show $c_{nt} \to t$ as $n\ti$. By
    \begin{align*}
        c_{nt} &= \frac{n^\thf\left(\D - p \right)}{m \zeta_n^\thf} \\
               &= \frac{tA}{m\zeta_n^\thf}\frac{F(\Ht_p + tAn^{-\thf}) - F(\Ht_p)}{tAn^{-\thf}},
    \end{align*}
    it follows, for $t > 0$ as $n \ti$,
    \begin{displaymath}
       c_{nt} \to \frac{tA F'(\Ht_p+)}{m \zeta^\thf}.
    \end{displaymath}
    Similarly, for $t < 0$ as $n\ti$,
    \begin{displaymath}
       c_{nt} \to \frac{tA F'(\Ht_p-)}{m \zeta^\thf}.
    \end{displaymath}
    Choosing
   \begin{displaymath}
       A = \frac{m\zeta^\thf}{F'(\Ht_p+)}
   \end{displaymath}
   if $t > 0$ and
   \begin{displaymath}
       A = \frac{m\zeta^\thf}{F'(\Ht_p-)}
   \end{displaymath}
   if $t < 0$, the claimed result follows.
\end{proof}

\section{Examples} \label{se:example}

\subsection{The Hodges-Lehmann estimator of location}
As an application of Theorem~\ref{th:main} (resp. Corollary~\ref{co:main}),
we deduce asymptotic normality of a slightly modified version of the
Hodges-Lehmann estimator~\cite{hodges63} of location and of a generalization.
The Hodges-Lehmann estimator is given by the median of all Walsh averages
\begin{displaymath}
    \frac{\xi_i + \xi_j}{2}, \quad 1 \leq i \leq j \leq n
\end{displaymath}
and estimates the location parameter associated with the one-sample Wilcoxon test, see e.g.~Hettmansperger~\cite{hettmansperger1998}.

If the Walsh averages with $i=j$ are dropped from the original definition of the Hodges-Lehmann estimator, this modification can be expressed easily as a $U$-quantile-statistic with $p = \thf$ and kernel
\begin{displaymath}
    h(x,y) = (x+y)/2.
\end{displaymath}

Let $\xi_1, \dots, \xi_n$ be i.i.d.~random variables with distribution $G$ and square integrable and continuous density $g$, symmetric about $0$ say, and $g(0) > 0$. Then $h(\xi_i,\xi_j)$, $1 \leq i < j \leq n$, with common d.f.~$F$ have continuous density $f$ and $f(0) > 0$ and thus Corollary \ref{co:main} can be applied directly. Clearly
\begin{displaymath}
    F(z) =\Prob{\frac{\xi_i+\xi_j}{2} \leq z} = \int \Prob{\xi_i \leq 2z - x}g(x)dx.
\end{displaymath}
Thus, by symmetry,
\begin{displaymath}
    f(0) = F'(0) = 2 \int g(x)^2 dx.
\end{displaymath}

The value of $\zeta$ is found easily by Remark \ref{re:conditioning}:
\begin{align*}
    \zeta + 1/4&= \int \left(\Prob{\xi_2 \leq x}\right)^2 g(x) dx  \\
          &= \E \left(G^2(X) \right) \\
          &= \E U^2
\end{align*}
for a standard uniformly distributed random variable $U$. Thus, $\zeta = 1/3 - 1/4 = 1/12 > 0$.

Corollary \ref{co:main} ensures asymptotic normality with mean 0 and variance $\sigma^2$
\begin{displaymath}
    \sigma^2 = \frac{m^2\zeta}{f^2(0)} = \left(12 \left(\int g^2(x) dx \right)^2\right)^{-1},
\end{displaymath}
which equals the corresponding result for the Hodges-Lehmann estimator, see Hettmansperger~\cite{hettmansperger1998}, p.~37.

In the same way, the asymptotic distributions for $U$-quantile-statistics with kernels
\begin{displaymath}
    h(x_1, \dots, x_m) = m^{-1} \sum_{i = 1}^m x_i, \quad m > 2,
\end{displaymath}
can be established by plugging
\begin{displaymath}
    f(0) = m \int \dots \int g(x_1 + \dots + x_{m-1}) g(x_1) \cdots g(x_{m-1}) dx_1 \cdots dx_{m-1}
\end{displaymath}
and
\begin{align*}
    \zeta &+ 1/4 \\
        = &\int \left(\int \dots \int G(x_1 + \dots + x_{m-1}) g(x_2)\cdots g(x_{m-1}) dx_1 \cdots dx_{m-1} \right)^2 g(x_1) dx_1
\end{align*}
into Corollary~\ref{co:main}.

\subsection{Median interpoint distance}
A geometric example of a $U$-quantile-statistic is given by the sample median $\hat{\theta}_n$ of all interpoint distances $\|\xi_i - \xi_j\|$ (with theoretical median $\theta$) of a sample of i.i.d.~points $\xi_1, \dots, \xi_n$ with continuous density $g$ in $\R^d$, $d \geq 1$. Distances are measured with respect to any fixed norm $\| \cdot \|$ on $\R^d$. The closed unit ball induced by this norm is denoted by $\Ball$ with surface $\Sphere$ and we write
\begin{align*}
    \{x + \theta \Ball\} &= \{y \in \R^d\!: \|y-x\| \leq \theta\} \ \text{and} \\
    \{x + \theta \Sphere\} &= \{y \in \R^d\!: \|y-x\| = \theta\}.
\end{align*}
The asymptotic normal distribution of $\hat{\theta}_n$ is established by Corollary~\ref{co:main} and Remark~\ref{re:conditioning}. The value of $\zeta$ is found via
\begin{align*}
    \zeta + 1/4 &= \int_{\R^d} \left(\Prob{\|\xi_1 - x\| \leq \theta}\right)^2 g(x) dx\\
                &= \int_{\R^d} \left(\Prob{\xi_1 \in \{x + \theta \Ball\}}\right)^2 g(x) dx \\
                &= \int_{\R^d} \left(\int_{\{x + \theta \Ball\}}g(y) dy\right)^2 g(x) dx,
\end{align*}
whereas the density $f$ of the random interpoint distance $\|\xi_1 - \xi_2 \|$ at $\theta$ is given by
\begin{align*}
    f(\theta)d\theta &= \Prob{\theta \leq \|\xi_1 - \xi_2 \| \leq \theta + d\theta} \\
                     &= \int_{\R^d} \Prob{\theta \leq \|\xi_1 - x\|  \leq \theta + d\theta}g(x) dx\\
                     &= \int_{\R^d} d\theta \left(\int_{\{x + \theta \Sphere\}}g(y)dy\right)g(x) dx,
\end{align*}
hence
\begin{displaymath}
    f(\theta) = \int_{\R^d} \left(\int_{\{x + \theta \Sphere\}}g(y)dy\right)g(x) dx.
\end{displaymath}

Corollary~\ref{co:main} ensures asymptotic normality.

\section*{Acknowledgements}
The author thanks Lutz Duembgen and Enkelejd Hashorva for their critical and fruitful review of the manuscript. You are great!


\newpage

\begin{thebibliography}{7}

\bibitem{bar:hol:jan92}
A.~D.~Barbour, L.~Holst and S.~Janson.
\newblock {\em Poisson Approximation}.
\newblock Clarendon Press, Oxford, 1992.

\bibitem{cal:jan78}
H.~Callaert and P.~Janssen.
\newblock  The Berry-Ess\'{e}en theorem for U-statistics.
\newblock {\em Ann.\ Statist.}, 6:417--421, 1978.

\bibitem{hettmansperger1998}
T.~Hettmansperger and J.~McKean.
\newblock {\em Robust Nonparametric Statistical Methods}.
\newblock Arnold, London, 1998.

\bibitem{hodges63}
L.~Hodges and L.~Lehmann
\newblock Estimates of location based on rank tests.
\newblock {\em Ann.\ Math.\ Statist.}, 34:598--611, 1963.

\bibitem{hoeffding48}
W.~Hoeffding.
\newblock A class of statistics with asymptotically normal distribution.
\newblock {\em Ann.\ Math.\ Statist.}, 19:293--325, 1948.

\bibitem{mayer08}
W.~Lao and M.~Mayer
\newblock U-Max-Statistics.
\newblock {\em Journal of Multivariate Analysis}, 99:2039--2052, 2008.

\bibitem{serfling80}
R.~Serfling.
\newblock {\em Approximation Theorems of Mathematical Statistics}.
\newblock Wiley, 1980.

\end{thebibliography}
\end{document}